\documentclass[11pt]{article}

\usepackage{latexsym,amsthm,amsmath,amssymb,url}

\newtheorem{lemma}{Lemma}

\newtheorem{theorem}{Theorem}

\newcommand{\OurProb}{\mbox{{\sc Min-DESC}}}
\newcommand{\NPh}{\mbox{{\cal NP}-hard}}

\newcommand{\DESC}{DESC}

\newcommand{\sdom}{\Rightarrow}

\begin{document}

\title{Parameterized Eulerian Strong Component Arc Deletion Problem on Tournaments}

\author{R. Crowston, G. Gutin, M. Jones and A. Yeo\\
Royal Holloway, University of London, UK \\ \url{{robert|gutin|markj|anders}@cs.rhul.ac.uk}}

\date{}
\maketitle

\begin{abstract}
\noindent
In the problem \OurProb, we are given a digraph $D$ and an integer $k$, and asked whether there exists a set $A'$ of at most $k$ arcs in $D$, such that if we remove the arcs of $A'$, in the resulting digraph every strong component is Eulerian. \OurProb{} is \NPh; Cechl\'{a}rov\'{a} and Schlotter (IPEC 2010) asked whether the problem is fixed-parameter tractable when parameterized by $k$.
We consider the subproblem of \OurProb{} when $D$ is a tournament. We show that this problem is fixed-parameter tractable with respect to $k$.
\end{abstract}

\section{Introduction}\label{section:intro}

For a digraph $D$, let $V(D)$ denote the vertices of $D$ and $A(D)$ the arcs of $D$.
Let $xy$ denote an arc from $x$ to $y$ and let $|D| = |V(D)|$.
Recall that the \emph{out-degree} of a vertex $x$ in a digraph $D$ is the number $d^+_D(x)$ of arcs of the form $xy$, for some vertex $y$, and the \emph{in-degree} of $x$ is the number $d^-_D(x)$ of arcs of the form $yx$.
$D$ is \emph{balanced} if for every vertex $x$ in $D$, $d^+_D(x) = d^-_D(x)$.
$D$ is \emph{regular} if for every pair $x,y$ of vertices, $d^+_D(x) = d^-_D(x) = d^+_D(y) = d^-_D(y).$
For a set $X$ of vertices in $D$, the subgraph $D[X]$ of $D$ \emph{induced} by $X$ is the digraph with $V(D[X]) = X$, $A(D[X]) = \{ xy: x, y \in X, xy \in A(D)\}$.

A digraph $D$ is \emph{strongly connected} if for every $x, y \in V(D)$ there is a directed path in $D$ from $x$ to $y$.
In particular, the digraph consisting of just one vertex is strongly connected.
A \emph{strong component} of $D$ is a maximal induced subgraph $C$ in $D$ that is strongly connected.
 $D$ is \emph{Eulerian} if there is a directed closed trail in $D$ that traverses every vertex of $D$ and uses every arc in $A(D)$ once.
Recall that $D$ is Eulerian if and only if $D$ is strongly connected and balanced.
For $X, Y \subseteq V(D)$ let $X \sdom Y$ denote the fact that all arcs between $X$ and $Y$ go from a vertex in $X$ to a vertex in $Y$. In
particular, we can write $X \sdom Y$ if there are no arcs between $X$ and $Y$.
 If $Y$ is a set of vertices in $D$ or a subgraph of $D$, for $x \in V(D)$ we let $d^+_Y(x)$ denote the number of arcs of the form $xy$, for some vertex $y$ in $Y$, and define $d^-_Y(x)$ similarly. For a positive integer $n$, let $[n]=\{1,\ldots,n\}.$ For more information on digraphs, see \cite{BangJensenGutin08}.

If $A'$ is an arc-set in $D$ then $D-A'$ denotes the subgraph of $D$ obtained by deleting all arcs of $A'$ from $D$.
A set $A' \subseteq A(D)$ is said to be a $\DESC{}$-set
(standing for {\em {\bf D}elete in order to obtain {\bf E}ulerian {\bf S}trong {\bf C}omponents})
if all strong components in $D-A'$ are Eulerian in $D-A'$. The size of a smallest $\DESC{}$-set
in $D$ is denoted by $desc(D)$.
We can now define the problem \OurProb{}.

\begin{quote}
  \OurProb{} \\ \nopagebreak
  \emph{Instance:} A digraph $D$ and a nonnegative integer $k$.\\
    \nopagebreak
\emph{Question:} Decide whether $desc(D) \leq k$.
 \end{quote}

Cechl\'{a}rov\'{a} and Schlotter \cite{CechSchlotter2010} introduced \OurProb{} in the context of housing markets, and asked if the problem is fixed-parameter tractable when parameterized by $k$.

Cygan et al. \cite{CygMarPilPilSch2011} study a number of related problems; in particular, the problem of finding a set of arcs $A'$ such that $D-A'$ is balanced, and the problem of finding a set of arcs $A'$ such that $D-A'$ is Eulerian.
They proved that the first problem is polynomial-time solvable and the second problem is fixed-parameter tractable when the parameter is $k$.
(Cygan et al. also studied related problems with graphs instead of digraphs, and deleting vertices instead of arcs or edges).

Note that unlike these problems, in \OurProb{} we do not require that $D-A'$ is balanced, since we allow arcs between strong components.
This makes it more difficult to say anything about what $D-A'$ must look like locally, as whether a given arc is between two strong components or not depends on the rest of the digraph.

The complexity of \OurProb{} parameterized by $k$ remains open.
In this paper, we consider \OurProb{} in the special case when $D$ is a tournament. This problem is still \NPh{} as proved in \cite{DESCarxiv}.
Since the proof is very long and of relatively little interest to the reader, we have decided to omit it from the paper.
In this paper we prove the following theorem.

\begin{theorem}\label{thmKernel} 
\OurProb{} for tournaments  parameterized by $k$ is fixed-parameter tractable and
has a kernel with at most $4k(4k+2)$ vertices.
\end{theorem}

A \emph{parameterized problem} is a subset $L\subseteq \Sigma^* \times
\mathbb{N}$ over a finite alphabet $\Sigma$. $L$ is
\emph{fixed-parameter tractable} if the membership of an instance
$(I,k)$ in $L$ can be decided in time
$f(k)|I|^{O(1)}$ where $f$ is a function of the
{\em parameter} $k$ only~\cite{DowneyFellows99,FlumGrohe06,Niedermeier06}.
Given a parameterized problem $L$,
a \emph{kernelization of $L$} is a polynomial-time
algorithm that maps an instance $(x,k)$ to an instance $(x',k')$ (the
\emph{kernel}) such that (i)~$(x,k)\in L$ if and only if
$(x',k')\in L$, (ii)~ $k'\leq h(k)$, and (iii)~$|x'|\leq g(k)$ for some
functions $h$ and $g$. The function $g(k)$ is called the {\em size} of the kernel.
It is well-known \cite{DowneyFellows99,FlumGrohe06,Niedermeier06} that a decidable parameterized problem $L$ is fixed-parameter
tractable if and only if it has a kernel. Polynomial-size kernels are of
main interest due to applications \cite{DowneyFellows99,FlumGrohe06,Niedermeier06}, but unfortunately many fixed-parameter problems
do not have such kernels unless  coNP$\subseteq$NP/poly, see, e.g., \cite{BDFH09,BTY09,DLS09}.

\section{Fixed-parameter Tractability Result}\label{sec:fpt}

In this section we will prove Theorem \ref{thmKernel}. To prove this result, consider a tournament $T$,
and assume $desc(T) \leq k$. Let $A' \subseteq A(T)$ be a $\DESC{}$-set for $T$ of size at most $k$ and let $T'=T-A'$.
Let $\succ$ be a linear ordering on the strong components of $T'$, such
that if $A \succ B$ then $A\sdom B$.
For sets of vertices $X$, $Y$, we let $X \succ Y$ denote the fact that
$A \succ B$, for any strong components $A, B$ in $T'$ with $V(A) \cap X \neq \emptyset$, $V(B) \cap Y \neq \emptyset$.

We firstly prove some properties of $T, T'$ and $A'$.
For sets of vertices $X$ and $Y$, let $d^*(X,Y)$ be the number of arcs in $A'$ with one end-vertex in $X$ and the other in $Y$ (in either direction). Let $d^*(X) = d^*(X,X)$, $d^*(x,Y) = d^*(\{x\},Y)$ and $d^*(x)=d^*(x,V(T))$.


\begin{lemma}\label{lem:deg1}
Let $A,B$ be two strong components in $T'$, with $A \succ B$, and let $W$ be
the (possibly empty) set  of vertices in components between $A$ and $B$ 
(i.e., $W$ is the maximal set of vertices disjoint from $A$ and $B$ such that $A \succ W \succ B$).
Suppose that $a \in A$ and $b \in B$. Then $d^+_T(a)-d^+_T(b) \ge \frac{|A|+|B|}{2}+|W|-(k+1)$.
\end{lemma}
\begin{proof}
Let $Z$ be the maximal set of vertices for which $B \succ Z$, and let $R$ be the maximal set of vertices for which $R \succ A$.

Note that if $X$ is the vertex set of a strong component in ${T'}$ and $x$ is a vertex in $X$,
then $d_{{T'}[X]}^+(x) = (|X|-1-d^*(x,X))/2$ and $$d_{{T'}[X]}^+(x) \le d_{T[X]}^+(x) \le d_{{T'}[X]}^+(x) + d^*(x,X).$$
Also note that if $X \succ Y$, then the arcs between $X$ and $Y$ in $A'$ will be exactly the arcs from $Y$ to $X$.
Hence, for all $a\in A$, $b\in B$ we have
$$d^+_T(a)-d^+_T(b) \ge \left( \frac{|A|-1-d^*(a,A)}{2} + |W| + |B|+|Z| - d^*(a,W\cup B\cup Z)\right) $$ $$- \left( \frac{|B|-1 + d^*(b,B)}{2}+|Z|+d^*(b,R\cup A\cup W) \right)$$
$$\ge \frac{|A|-|B|}{2}+|W|+|B|-d^*(a,A) - d^*(b,B)-d^*(a,W\cup B\cup Z)-d^*(b,R\cup A\cup W)$$
$$\ge \frac{|A|+|B|}{2}+|W|-(k+1).$$
The last inequality holds since $d^*(a,A)+d^*(b,B)+d^*(a,W\cup B\cup Z)+d^*(b,R\cup A\cup W)\le k+1$,
as there are at most $k$ deleted arcs, with each deleted arc counted at most once,
with a possible exception of an arc from $a$ to $b$, which would be counted twice.
 \end{proof}

\begin{lemma}\label{lem:split}
Suppose there exists $X \subseteq V(T)$ such that $|X| \ge 4k+3$ and
$ \max_{x\in X} d^+_T(x) - \min_{x\in X} d^+_T(x) \le k.$
Then $X$ belongs to one strong component in ${T'}$.
\end{lemma}

 \begin{proof}
 For the sake of contradiction assume that the lemma does not hold.
Let $A,B,W$ be disjoint sets of vertices such that $A \succ W \succ B$, $X \subseteq A \cup W \cup B$, and $A$ and $B$ are the vertex sets of strong components in $T'$.
By assumption,
$|A|+|B|+|W| \ge |X| \ge 4k+3$. Applying Lemma~\ref{lem:deg1}, for any $a\in A\cap X$, $b\in B\cap X$, we have
$d^+_T(a)-d^+_T(b)\ge (4k+3)/2-(k+1) > k,$
which contradicts our assumption.
\end{proof}

\begin{lemma}\label{lem:same}
 If $x$ and $y$ belong to the same strong component in ${T'}$ then \newline
$|d^+_T(x)-d^+_T(y)| \le k.$
\end{lemma}
\begin{proof} 
Let $H$ be the strong component of $T$ containing $x$ and $y$, and let $V=V(T)$ and $U=V(H).$
Form a digraph $Q$ from $T'$ by reinserting deleted arcs between $V$ and $V\setminus U$ and then reorienting them.
Observe that $d^+_{Q[V\backslash U]}(x) = d^+_{Q[V\backslash U]}(y)$, and thus $|d^+_{T[V\backslash U]}(x)-d^+_{T[V\backslash U]}(y)| \le d^*(U, V\backslash U)$.
Also recall (see the second paragraph in the proof of Lemma \ref{lem:deg1}) that $\frac{|U|-1-d^*(u,U)}{2} \le d^+_{T[U]}(u) \le \frac{|U|-1+d^*(u,U)}{2}$ for any $u \in U$, and thus $|d^+_{T[U]}(x)-d^+_{T[U]}(y)| \le \frac{d^*(x,U) + d^*(y,U)}{2}$. Hence, \newline
$|d^+_T(x)-d^+_T(y)| \le d^*(U, V\backslash U) + \frac{d^*(x,U) + d^*(y,U)}{2} \le k.$
\end{proof}

We also require the following lemma, which is proved in \cite{CygMarPilPilSch2011}.

\begin{lemma} \label{polyAlg}
 Given a directed graph $D$, we can in polynomial time find a set $A' \subseteq A(D)$ of minimal size such that $D-A'$ is balanced.
\end{lemma}

We are now ready to provide a proof for Theorem  \ref{thmKernel}.

\begin{proof}[Proof of Theorem \ref{thmKernel}]
We may assume that $(T,k)$ is a {\sc Yes}-instance of \OurProb. Indeed, our proof below always replaces $(T,k)$ with an equivalent instance and so if $(T,k)$ is a {\sc No}-instance, we will either arrive at a negative value of the parameter or at the kernel which is a {\sc No}-instance.

First note that if $T$ already contains a Eulerian strong component $C$, we may remove $C$ from $T$, as any minimal \DESC{}-set for $T$ is also a minimal \DESC{}-set for $T-C$, and vice versa.


Let $V(T)=\{v_1,v_2,\ldots,v_n\}$ and assume without loss of generality that $d^+(v_1) \geq d^+(v_2) \geq \cdots \geq d^+(v_n)$. We will now form subsets of these vertices, $Q_i$ and $W_i$ such that the indices of the vertices are exactly an interval.

Partition the vertices into sets $Q_1,\ldots ,Q_s$ as follows.
 Let $x_1 = 1$. Let $y_i = \max\{r: d^+(v_{x_i}) - d^+(v_r) \le k \}$. If $y_i < n$, let $x_{i+1} = y_i + 1$. For each $i$, let $Q_i = \{ v_{x_i}, v_{x_i + 1}, \ldots , v_{y_i}\}$.
Observe that $d^+(v_{x_i}) - d^+(v_{r}) > k$ for each $r\ge x_{i+1}$.


We first deal with the case in which $|Q_i| \ge 4k +3$, for some $Q_i$.
Let $z_i = \min \{r: d^+(v_r) - d^+(v_{y_i}) \le k\}$, and let $W_i = \{ v_{z_i}, v_{z_i + 1}, \dots v_{y_i} \}$.
Note that $W_i \supseteq Q_i$ and so $|W_i| \ge 4k+3$, and thus Lemma \ref{lem:split} implies that $W_i$ belongs to one strong component in ${T'}$.
Thus, for any vertex $v$ outside of $W_i$, either $d^+(v)-d^+(v_{y_i}) > k$ or $d^+(v_{x_i})-d^+(v) > k$, and so Lemma \ref{lem:same} implies that $v$ is not in the same strong component as $W_i$ in ${T'}$. Therefore $W_i$ is exactly the vertices of one strong component in ${T'}$.
Lemma \ref{lem:deg1} implies that $\{v_r\} \succ W_i$ for any $r < z_i$ and $W_i \succ \{v_r\}$ for any $r > y_i$. Hence, we delete any arc from $W_i$ to $v_r$ when $r < z_i$, and any arc from $v_r$ to $W_i$ when $r > y$ and reduce the parameter appropriately.


$W_i$ is now a strong component in the resulting digraph $T$. Let $p$ be the number of arcs of $A'$ in $T[W_i]$. Apply the polynomial algorithm of Lemma \ref{polyAlg} to remove the minimum number $p'$ of arcs from within $T[W_i]$ such that every vertex in $T[W_i]$ becomes balanced. Observe that $p'\le p$. After the removal of arcs, $T[W_i]$ will be a disjoint union of Eulerian digraphs and so it can be removed from $T$. We will also decrease $k$ by $p'$. Repeat the above process until  $|Q_i| \le 4k+2$ for every $Q_i$.

%


Next, observe that if there is an arc
 from $a\in Q_j$ to $b\in Q_i$, with $j\ge i+2$, then it must belong to $A'$,
since by Lemma~\ref{lem:same}, $a$ and $b$ are in different strong components of $T'$, and by Lemma~\ref{lem:deg1}, $a \succ b$ is not possible.
Therefore we delete any such arc from $T$ and put it in a set $B$ of arcs ($B$ is empty initially) and decrease $k$ appropriately. 
Finally, remove any strong Eulerian component from $T$.

Now we will show that the number of vertices still in $T$ is at most $4k(4k+2)$.
For every strong component $C$ in $T$,
let $l = \min \{i: Q_i \cap V(C) \neq \emptyset \}$ and $r = \max \{i: Q_i \cap V(C) \neq \emptyset \}$.
Consider a directed path $P$ from $Q_r$ to $Q_l$ (such a path exists as $C$ is strongly connected). 
Such a path must have an arc going from $Q_{i+1}$ to $Q_i$ for every $l \le i < r$.

Suppose $r \ge l + 2$ and consider an index $j$ such that $l\le j\le r-2$. Let $v_s$ be the last vertex of $Q_{j+2}$ on $P$ and $v_t$ the first vertex of $Q_j$ on $P$. It follows from the definition of $Q_i$'s that $|d^+(v_t)-d^+(v_s)|>k$ and so, by Lemma~\ref{lem:same}, $v_s$ and $v_t$ must belong to different strong components of $T'$. Thus, at least one arc in the subpath of $P$ from $v_s$ to $v_t$ must be removed to form $T'$.
This implies that at least $(r-l-1)/2$ arcs inside $C$ must be removed.

If $r < l+2$ then still at least one arc inside $C$ needs to be removed, as otherwise $C$ 
would be a Eulerian strong component and we would have removed it.

Therefore we need to remove at least $s/4$ arcs, where $s$ is the total number of sets $Q_i$. Hence we may assume that $s \le 4k$.
Recall that $|Q_i|\le 4k+2$ for each $i$ and, thus, the number of vertices still in $T$ is at most $4k(4k+2)$. 

Note that $T$ may not be a tournament, but we can turn it into one by adding all arcs from $B$ which have both vertices in $T$. We increase $k$ by the number of added arcs. Observe that the bound $4k(4k+2)$ on $|V(T)|$ remains valid.
 \end{proof}

\vspace{0.5cm}

\paragraph{Acknowledgments} We are grateful to the referees for several very useful suggestions and remarks.


\end{document}